\documentclass[11pt]{article}
\usepackage[symbol]{footmisc}
\renewcommand{\thefootnote}{\fnsymbol{footnote}}

\date{}
\usepackage{amsmath,amssymb,amsthm}
\usepackage{bbm}
\usepackage{framed}
\usepackage{tabu}
\usepackage{pgf}
\usepackage{tikz}
\usepackage{verbatim}
\usepackage{braket}
\usepackage{enumerate}
\usepackage{hyperref}
\usepackage{cleveref}
\usepackage{thmtools}
\usepackage{thm-restate}
\usepackage[margin=1in]{geometry}

\hypersetup{
	colorlinks,
	linkcolor={blue!100!black},
	citecolor={red!100!black},
}

\DeclareMathAlphabet{\pazocal}{OMS}{zplm}{m}{n}

\newtheorem{theorem}{Theorem}[section]
\newtheorem{definition}[theorem]{Definition}

\newtheorem{fact}[theorem]{Fact}

\newtheorem{lemma}[theorem]{Lemma}

\newtheorem{claim}[theorem]{Claim}

\newcommand{\E}{\mathbb{E}}

\newcommand{\abs}[1]{\left| #1 \right|}
\newcommand{\vabs}[1]{\left\| #1 \right\|}
\newcommand{\abra}[1]{\left\langle #1 \right\rangle}
\newcommand{\pbra}[1]{\left( #1 \right)}
\newcommand{\sbra}[1]{\left[ #1 \right]}

\renewcommand{\mid}{\,\middle\vert\,}

\newcommand{\SU}{\mathsf{SU}}

\newcommand{\Haar}{\pazocal{H}}

\def\01{\{-1,1\}}
\DeclareMathOperator{\Tr}{Tr}

\newcommand{\indi}{\mathbbm{1}}

\newcommand{\Rtwo}{\textsf{R2}}

\newcommand{\SMP}{\textsf{SMP}}
\newcommand{\ABCD}{\textsf{ABCD}}

\newcommand{\id}{\mathbb{I}}

\newcommand{\Cbb}{\mathbb{C}}
\newcommand{\Nbb}{\mathbb{N}}
\newcommand{\Rbb}{\mathbb{R}}

\newcommand{\Lcal}{\mathcal{L}}

\newcommand{\Pcal}{\mathcal{P}}
\newcommand{\Rcal}{\mathcal{R}}

\newcommand{\Gcal}{\mathcal{G}}

\newcommand{\mud}{\mu_{\mathsf{diag}}}

\newcommand{\DQC}{\mathsf{DQC}_1}
\newcommand{\BPP}{\mathsf{BPP}}

\author{Srinivasan Arunachalam\thanks{	IBM Quantum, Almaden Research Center. Email: \href{Srinivasan.Arunachalam@ibm.com}{srinivasan.arunachalam@ibm.com}}
	\and Uma Girish\thanks{  Princeton University. Email:  \href{mailto:ugirish@cs.princeton.edu}{ugirish@cs.princeton.edu} }
	\and Noam Lifshitz\thanks{	Hebrew University of Jerusalem. Email: \href{noamlifshitz@gmail.com}{noamlifshitz@gmail.com}}}

\newcommand{\srini}[1]{{\textcolor{blue}{Srini}}}

\begin{document}
	\title{One Clean Qubit Suffices for Quantum  Communication~Advantage}
	\maketitle
	
	\begin{abstract}  We study the one-clean-qubit model of quantum communication where one qubit is in a pure state and all other qubits are maximally mixed. We demonstrate a partial function that has a quantum protocol of cost $O(\log N)$ in this model, however, every interactive randomized protocol has cost $\Omega(\sqrt{N})$, settling a conjecture of Klauck and Lim. In contrast, all prior quantum versus classical communication separations required at least $\Omega(\log N)$ clean qubits. The function demonstrating our separation also has an  efficient protocol in the quantum-simultaneous-with-entanglement model of cost $O(\log N)$. We thus recover the state-of-the-art separations between quantum and classical communication complexity. Our proof is based on a recent hypercontractivity inequality introduced by Ellis, Kindler, Lifshitz, and Minzer, in conjunction with tools from the representation theory of compact Lie groups.
	\end{abstract}

\section{Introduction}
\renewcommand*{\thefootnote}{\arabic{footnote}}

A central goal in complexity theory is to understand the power of different computational resources. In the past four decades, communication complexity has provided a successful toolbox to establish several results in theoretical computer science in circuit complexity~\cite{KarchmerW90,karchmer1995super}, streaming algorithms~\cite{kapralov2014streaming}, property testing~\cite{blais2012property}, extension complexity~\cite{fiorini2015exponential}, data structures~\cite{miltersen1995data}, proof complexity~\cite{huynh2012virtue}. 
In the standard two-player model of communication complexity introduced by Yao~\cite{yao1979some} there are two parties Alice and Bob whose goal is to compute a partial function $F:\mathcal{X}\times \mathcal{Y}\rightarrow \{-1,1,\star\}$. Alice receives $x\in \mathcal{X}$  and Bob receives $y\in \mathcal{Y}$  and their goal is to compute $F(x,y)$ for all $(x,y)\in F^{-1}(1)\cup F^{-1}(-1)$, while minimizing the number of bits of communication. One variant of this is when the players are allowed to send quantum messages. Quantum versus classical separations in communication complexity have a long and rich history. In a sequence of works~\cite{buhrman1998quantum,bar2008exponential,buhrman2001quantum,raz1999exponential,gavinsky2006bounded,gavinsky2007exponential,klartagregev,gavinsky2020bare,GRT22}, it has been shown that quantum communication can exponentially outperform classical communication.\footnote{Proving such separations for total functions is a major open question, however, we know several examples of partial functions such that quantum provides provable exponential speedups. Total functions are defined on all possible inputs, while partial functions are defined on a subset of inputs.}  The state-of-the-art result among these is due to~\cite{Gav16,GRT22}, who give a separation between quantum simultaneous communication complexity (where Alice and Bob share entanglement) and interactive randomized. This result subsumes most previous results and also shows that a rather weak and restricted model of quantum communication (simultaneous with entanglement) can exponentially outperform a rather strong classical model (interactive randomized). One of the benefits of proving quantum versus classical separations in communication complexity is that they are \emph{unconditional}. The motivation for our work is two-fold:

\textbf{Near-term implementations.} Given that we are finally in an era of small noisy quantum devices, there have been proposals to use these communication separations to show experimental demonstrations of quantum advantage. To this end, Kumar et al.~\cite{kumar2019experimental} experimentally \emph{demonstrated} a quantum communication advantage for the Hidden-matching problem defined by Bar-Yossef et al.~\cite{bar2008exponential}. More recently, Aaronson coined the term ``\emph{quantum information supremacy}" wherein the goal is to show a task is solvable using a quantum resources that is exponentially 
 more efficient than classical resources, with the benefit that this quantum advantage would be \emph{unconditional} unlike sampling-based proposals. Aaronson et al.~\cite{ABK23} again used the communication problem of~\cite{bar2008exponential} and showed a separation between quantum and classical complexity classes, and proposed an experimental implementation~\cite{aar23}. Inspired by these recent works, we ask
\begin{quote}
 \emph{What is the minimum quantum resource sufficient for a quantum communication speedup?}
 \end{quote}
\textbf{$\DQC$ versus $\BPP$.} Knill and Laflamme~\cite{KL98} introduced the  one-clean qubit model of quantum computation, also known as $\DQC$. 
In this model, there is one qubit in a pure state and all other qubits are maximally mixed. The motivation of this model is two-fold: $(i)$ The idea is to study the power of models of quantum computing in which the quantum memory is weak, but the control of this memory is good, in contrast  to studying quantum
computation, where the underlying memory is good, but the control is weak (such as Boson sampling) $(ii)$ The primary motivation of~\cite{KL98} was the NMR approach to quantum computing where the initial state may be highly mixed.  Despite how noisy these states are, $\DQC$ is powerful and provides exponential speedups compared to the best known classical models~\cite{SJ08,KL98,CM18}.
It was shown~\cite{MFJ14} that $\DQC$ is not efficiently classically simulable, unless the polynomial hierarchy collapses to the second level. While these results provide strong evidence that $\DQC$ can exponentially outperform classical computation, proving this unconditionally has been a long-standing open question in complexity theory (in particular, what is the relation between $\DQC$ and  $\BPP$).  All known hardness results rely on complexity theoretic assumptions.  When it comes to unconditional separations between quantum and classical, there are very few settings where quantum models provably exponentially outperform classical models. Communication complexity is a striking example of such a setting and a natural question~is 
\begin{quote}
\begin{center}
 \emph{Is $\DQC\subseteq \BPP$ in the communication world unconditionally?}
 \end{center}
 \end{quote}
In the $\DQC$ model of communication (first defined by Klauck and Lim~\cite{KL19}), Alice and Bob exchange quantum states such that the first qubit is in a pure state and all other qubits are maximally mixed. The players have no additional private memory, they simply take turns applying unitary operators on the state. (See~\Cref{sec:clean_qubit_model} for a formal definition.)  This is a rather restrictive model of quantum communication; all the aforementioned quantum versus classical separations require the quantum protocol to have at least $\Omega( \log N)$ clean qubits. They proposed a natural communication problem that is solvable using only one clean qubit. We call this the \ABCD~problem. Here, $A,B,C,D$ are $N\times N$ special unitary matrices, Alice gets as input $A,C$ explicitly and Bob gets as input $B,D$ explicitly and their goal is to decide if $\Tr(ABCD)\ge 0.9N$ or $\Tr(ABCD)\le 0.1N$ promised one of them is the case. 
 This problem was shown to have a protocol with $O(\log N)$ qubits in the one-clean-qubit model of quantum communication~\cite{KL19}.\footnote{\label{note1} Although they state their protocol for the $\mathsf{ABC}$ problem, the protocol trivially extends to the \ABCD~problem.} They also conjectured  that the interactive randomized communication complexity of this problem is $\Omega(\sqrt{N})$. The main contribution of this paper is to prove the conjecture of Klauck and Lim~\cite{KL19}. 
 
\paragraph*{Main Result.} 
We show that the $\ABCD$ problem can be computed with cost $O(\log N)$ with just \emph{one-clean qubit}, however, every interactive randomized protocol has cost $\Omega(\sqrt{N})$. This separates $\DQC$ and $\BPP$ unconditionally in the communication world. As far as we are aware, all prior quantum versus classical communication separations required at least $\Omega(\log N)$ clean qubits.

We also study the communication complexity of the \ABCD~problem in the quantum-simultaneous-with-entanglement model. In this model, Alice and Bob share entanglement, each apply a quantum operation on their part of the shared state and send everything to Charlie, who applies a projective measurement and announces the outcome as the answer. Interestingly, there is a protocol of cost $O(\log N)$ in this model for the \ABCD~problem. As a result, we show an exponential separation between quantum-simultaneous-with-entanglement and interactive randomized communication complexity and thus recover many of the best known separations, including~\cite{gavinsky2019quantum,GRT22}. Our quantum simultaneous protocol for the $\ABCD$~problem has the additional nice property that it is an entangled-\emph{fingerprinting} protocol, i.e., a type of simultaneous protocol where Charlie essentially just performs a swap test. We describe this in more detail now. 

Gavinsky et al.~\cite{GKW06} introduced the quantum \emph{fingerprinting} model in the $\SMP$ model: here Alice and Bob on input $x,y$ respectively, send $U_x\ket{0^n},V_y\ket{0^n}$ to Charlie who performs a swap test between $U_x\ket{0^n}$ and $V_y\ket{0^n}$. They repeat this process a few times before Charlie obtains the swap-test statistics and computes  $F$ on inputs $x,y$. Surprisingly, it was shown~\cite{GKW06} that this model is efficiently simulable in the classical randomized simultaneous model of communication, thereby showing that quantum states are no stronger than classical states for the fingerprinting model. We consider the \emph{entangled-fingerprinting} model where Alice and Bob share a few EPR pairs and on input $x,y$, apply $U_x,V_y$ on their part of the shared state and send it to Charlie, who still performs a swap test between Alice's $A$-register and Bob's $B$-register. (See~\Cref{sec:entangled_fingerprinting} for a formal definition of this model.) They repeat this in parallel with a fresh copy of $\ket{\psi}_{AB}$ and based on the swap test statistics, Charlie computes $F(x,y)$. 
A natural question is, \emph{are entangled fingerprints stronger than just randomized fingerprints, and if so how much stronger?} In this work, we show that in contrast to the standard fingerprinting model, the entangled-fingerprinting model can \emph{exponentially} outperform randomized fingerprinting and even outperform the strongest \emph{interactive} classical model of communication. 

\textbf{Techniques.}
Given the definition of the $\ABCD$ problem, compact Lie groups such as the special unitary group $\SU(N)$ arise naturally. Our work draws inspiration from the study of quasirandom groups. A group $G$ is said to be $D$-\emph{quasirandom} if every $D$-dimensional representation $\rho\colon G\to \mathrm{GL}_D(\mathbb{C})$ is trivial, i.e., constantly equal to the identity.
 Group quasirandomness plays a central role in number theory~\cite{sarnak1991bounds}, group theory~\cite{bourgain2008uniform} and combinatorics~\cite{gowers}. Gowers and Viola~\cite{gowers2015communication}  showcased how quasirandomness transcends its origins in pure mathematics, employing it as a pivotal tool in proving lower bounds for a variety of communication protocols over groups. For our purpose, it turns out that the quasirandomness of $\SU(N)$ alone is not sufficient. To overcome this, we instead introduce a set of new deep mathematical tools and concepts from the study of product free sets in $\SU(n)$ \cite{keevash2022largest, EKLM23} into communication complexity theory. 

Our proof uses representation theory, Fourier analysis and hypercontractivity on $\SU(N)$. Fourier analysis on the Boolean cube and level-$k$ inequalities have been important to prove  quantum vs.~classical separations~\cite{gavinsky2007exponential,GRT22,gavinsky2019quantum,montanaro2010new,doriguello2020exponential,ben2008hypercontractive,shi2012limits}. However, our problem is defined on $\SU(N)$ and it is unclear if the Fourier-analytic techniques on the  Boolean cube extend to the special unitary group $\SU(N)$. A recent breakthrough work~\cite{EKLM23} studied product-free sets in quasirandom compact Lie groups, and showed hypercontractive inequalities for $\SU(N)$. Although not immediate, their hypercontractive inequality is pivotal for the $\Rtwo$ lower bound. Our work appears to be the first application of this inequality to quantum computing and we believe that hypercontractivity on $\SU(N)$ will be of great interest to a broader quantum~audience. 

 \subsection{Main Theorem}

\begin{definition}[$\ABCD$ problem]
    Let $A,B,C,D$ be $N\times N$ special unitary matrices. Alice is given $A,C$ explicitly and Bob is given $B,D$ respectively. Their goal is to  output 1 if $\Tr(ABCD)\ge 0.9N$ and 0 if $\Tr(ABCD)\le 0.1N$, promised that one of these is true. 
\end{definition}
 Our main theorem is as follows.
\begin{theorem}\label{theorem:main_overall_theorem}
The $\ABCD$ problem has communication complexity
\begin{enumerate}
    \item\label{item1}   $O(\log N)$ in the one-clean qubit quantum model.
    \item\label{item2} $O(\log N)$ in the quantum-simultaneous-with-entanglement model.
    \item\label{item3}  $\Omega(\sqrt{N})$ in the interactive classical randomized  model.
\end{enumerate}
\end{theorem}
 We remark that the classical lower bound is tight as shown in~\cite{KL19}.\footnote{Although they state their protocol for the $\mathsf{ABC}$ problem, the protocol trivially extends to the \ABCD~problem.}






\subsection{Proof Overview}

We now describe our classical lower bound for the $\ABCD$ problem, the main technical contribution. Our proof is based on a combination of three main ingredients. The first two of which are dimensional lower bounds for irreducible representations and formulas for convolution; these involve the representation theory of compact Lie groups. The third ingredient is the level-$d$ inequality of \cite{EKLM23}, which shows that the Fourier spectrum of an indicator of a small subset of $\SU(n)$ is concentrated on the high dimensional representations. 

\paragraph*{Translating our lower bound to an analytic statement.} 

To prove our lower bound, we define two distributions: Alice's inputs $A,C$ and Bob's input $B$ are chosen uniformly from $\SU(N)$ and in the
\textsc{yes} distribution, Bob is given $D=(ABC)^{-1}$ and  in the \textsc{no} distribution, Bob is given uniformly random $D$ from $\SU(N)$. We show that distinguishing between these two distributions~requires $\Omega(\sqrt{N})$ communication. It is well-known, if we consider the matrix with rows and columns indexed by inputs of Alice and Bob respectively, then a classical cost-$c$ communication protocol partitions this matrix into $2^c$ combinatorial rectangles. So, for a cost $c$, a typical rectangle in this partition has measure $\approx 2^{-c}$. Thus, to prove that a certain function requires $\Omega(\sqrt{N})$ classical communication cost, it suffices to show that rectangles of measure $\approx 2^{-\sqrt{N}}$ cannot distinguish the \textsc{yes} and \textsc{no} instances of the function with sufficient advantage.  Translating this to our setting, the main technical heart of our paper is 
the following lemma about large rectangles $f\times g$ in $\SU(N)^2\times \SU(N)^2$ (think of $f \subseteq \SU(n)^2$ (resp.~$g$) as indicators of Alice (resp.~Bobs) inputs in that~rectangle).
\begin{lemma}[Main Lemma] \label{lemma:main_lemma_intro}
Let $f,g:\SU(N)^2\to \{0,1\}$ be indicator functions such that $\E[f]=\alpha$ and $\E[g]=\beta$ for $\alpha,\beta \ge e^{-c'\sqrt{N}}$ for a sufficiently small global constant $c'>0$. Then,
\[ \E\sbra{f(A,C)\cdot g(B,(ABC)^{-1})}\approx \alpha\beta\cdot (1\pm 0.1)\]
where $A,B,C,D$ are chosen independently according to the Haar probability measure over $\SU(N)$.
\end{lemma}
We now sketch the proof of this lemma. 
\paragraph{Applying convolution formulas from nonabelian Fourier analysis}

Firstly, we study the difference  between $\E[f(A,C)\cdot g(B,(ABC)^{-1})]$ and $\alpha\beta:=\E[f(A,C)\cdot g(B,D)]$ and derive an expression for this in terms of the product of Fourier coefficients of $f$ and $g$. 
\begin{claim}[Main Claim]\label{claim:fourier_expansion_intro}
    Let $G=\SU(N)$ and   $f,g:G\times G\to \{0,1\}$ be the indicator functions with $\alpha=\E[f]$ and $\beta=\E[g]$. Then, 
$$
\Delta:= \E\sbra{f(A,C)g(B,(ABC)^{-1})}-  \alpha\beta=\sum_{\emptyset \neq\pi\in \widehat{G}}\frac{1}{\dim(\pi)} \abra{ \widehat{f}(\pi,\pi), \widehat{g}(\pi,\pi)}.
$$
where $\widehat{G}$ denotes the equivalence class of irreps of $G$ and $\widehat{f}(\pi,\pi)\in \Cbb^{\dim(\pi,\pi)\times \dim(\pi,\pi)}$ denotes the (matrix) Fourier coefficient corresponding to $\pi$. 
\end{claim}
The proof of this uses Fourier analysis, especially facts about Fourier coefficients of convolutions of functions. Loosely speaking,
\begin{itemize}
    \item Taking an expectation of $f\times g$ over the distribution induced by $(A,B,C,(ABC)^{-1})$ has the effect of taking the convolution of $f$ and $g$, which in the Fourier basis, translates to a taking a product of Fourier coefficients of $f$ and $g$, divided by the square root of the dimension.
\item The term $(ABC)^{-1}$ has the effect of zeroing out Fourier coefficients corresponding to $(\pi,\sigma)$ where $\pi$ and $\sigma$ are inequivalent representations of $G$.
\end{itemize}

\subsection*{Utilizing the Fourier concentration on the high dimensions} 
We now describe how to upper bound the R.H.S. of~\Cref{claim:fourier_expansion_intro}. To do this, we will use the degree-decomposition of $f$ and $g$. It turns out that the space $L^2(G)=\{f:G\to \Cbb,\E[\abs{f}^2]<\infty\}$ can be expressed as $\oplus_{d\in  \Nbb} V_{d}\oplus V_0$, where $V_0$ consists of constant functions, $V_d$ essentially captures polynomials of ``pure-degree'' $d$, furthermore, each $V_d$ is a sub-representation of $G$~\cite{EKLM23}. We group the terms in the R.H.S. of~\Cref{claim:fourier_expansion_intro} based on this degree decomposition to obtain
\[\Delta= \sum_{d=1}^\infty \sum_{\pi\in\widehat{V_d}} \frac{1}{\dim(\pi)} \abra{ \widehat{f}(\pi,\pi), \widehat{g}(\pi,\pi))}\]
We now apply Cauchy-Schwarz to upper bound $\abra{ \widehat{f}(\pi,\pi), \widehat{g}(\pi,\pi))}$ by $\|\widehat{f}(\pi,\pi)\|_{2} \cdot \|\widehat{g}(\pi,\pi)\|_{2}$. 
We again apply Cauchy-Schwarz over terms $\pi\in \widehat{V_d}$ to obtain
\[\Delta \le   \sum_{d=1}^\infty  \sqrt{\sum_{\pi\in\widehat{V_d}}\|\widehat{f}(\pi,\pi)\|^2} \cdot  \sqrt{\sum_{\pi\in\widehat{V_d}}\|\widehat{g}(\pi,\pi)\|^2} \cdot \max_{\pi\in\widehat{V_d}}\pbra{\frac{1}{\dim(\pi)}}
   \]
Observe that $\widehat{f}(\pi,\pi)$ and $\widehat{g}(\pi,\pi)$ correspond to the degree-$2d$ component. Thus, 
\begin{align}\label{eq:intro} 
    \Delta &\le  \sum_{d=1}^\infty  \vabs{f^{=2d}}\cdot \vabs{g^{=2d}}\cdot \pbra{\min_{\pi\in\widehat{V_d}}{\dim(\pi)}}^{-1}  \end{align}
where $f^{=2d},g^{=2d}$ denote the projection of $f,g$ onto the degree $2d$ part. We now use the main results of~\cite{EKLM23}. Two important contributions of~\cite{EKLM23} are the following. Firstly, the dimensions of irreps of $\widehat{V_{d}}$ grow fast, roughly as $\gtrapprox N^{d}$; secondly, an analogue of the level-$k$ inequality holds for $\SU(N)$ and its variants:
\begin{lemma}[Implied by~\cite{EKLM23}]There exists universal constants $c,C>0$ such that the following holds. Let $f:\SU(N)^2\to \{0,1\}$,  $\alpha=\E[f]$ and $d\le \min\{c\sqrt{N},\log(1/\alpha)/2\}$. Then
\[
\vabs{f^{=d}}_2^2 \le \pbra{C/d}^d   \alpha^2\log^d(1/\alpha). 
\]
\end{lemma}
Since $\E[f],\E[g]\ge e^{-c'\sqrt{N}}$ for a sufficiently small constant $c'$, this lemma essentially implies that
\[\vabs{f^{=2d}}_2^2,\vabs{g^{=2d}}_2^2\ll  \alpha\beta\cdot N^d\cdot 11^{-d}.\]
As mentioned earlier, we have $\min_{\pi\in \widehat{V_d}}\dim(\pi)\gtrapprox N^d$ and thus $\min_{\pi\in \widehat{V_d}}{\dim(\pi)}$ grows much faster than $\vabs{f^{=2d}}\cdot \vabs{g^{=2d}}$. Plugging this in Eq.~\eqref{eq:intro} implies the desired result:
\[\E[f(A,C)\cdot g(B,(ABC)^{-1})]-\alpha\beta \triangleq \Delta\ll \sum_{d=1}^\infty \alpha\beta\cdot N^d \cdot 11^{-d} \cdot 1/N^d \le \alpha \beta\cdot \sum_{d=1}^{\infty}11^{-d}\le \alpha\beta/10.\] 
Using standard techniques in communication complexity, we use the above inequality to show that every protocol of cost $\ll \sqrt{N}$ succeeds in solving the $\ABCD$ problem with probability $\leq 1/10$, completing the proof sketch. 



\paragraph*{Organization.} We describe notation in~\Cref{sec:notation} and some important results about special unitary matrices in~\Cref{sec:properties_sun}. We describe the quantum communication models in~\Cref{sec:main_proof_quantum}. As mentioned before, \Cref{item1} in~\Cref{theorem:main_overall_theorem} was essentially proved in~\cite{KL19}, but we reprove it in~\Cref{sec:clean_qubit_upper_bound} for completeness. We prove~\Cref{item2} in~\Cref{sec:quantum_upper_bound} and~\Cref{item3} in~\Cref{sec:classical_lower_bound}. 


\paragraph*{Acknowledgements.} This work was done in part while the authors were visiting the Simons Institute for the Theory of Computing. We thank Vojtech Havlicek, Tarun Kathuria, Ran Raz and Makrand Sinha for the many valuable discussions. 

\section{Notation}
\label{sec:notation}
Let $\Nbb=\{1,2,\ldots,\}$. We use $\indi[E]$ to denote the indicator of an event $E$.  Let $N=2^n$ for $n\in \Nbb$.  For a vector space $V$, we use $\Lcal(V)$ to denote the space of \emph{endomorphisms} of $V$ (i.e., set of linear maps from $V$ onto itself) and $\Gcal\Lcal(V)$ to denote the set of invertible endomorphisms in $\Lcal(V)$. For every compact group $G$ equipped with a Haar measure, we use $L^2(G)$ to denote the space of all square-integrable functions acting on $G$ quotiented by the equivalence relation $f\sim g$  if $f$ equals $g$ almost everywhere with respect to $\mu$. Mathematically, we write $L^2(G)$ as 
$$
L^2(G)=\{ f:G\to \Cbb \hspace{0.5mm} \vert \hspace{0.5mm} \E[|f|^2]<\infty \}/\sim,
$$
where the expectation is with respect to the Haar measure. One can also view $L^2(G)$ as a Hilbert space equipped with a natural inner product 
$$
\langle f_1,f_2\rangle = \E_g[f_1(g)\overline{f_2(g)}]
$$ 
for $f_1,f_2\in L^2(G)$, where $g$ is Haar random.

\section{Properties of Special Unitary Matrices}
\label{sec:properties_sun}

\paragraph*{Special Unitary Group.} We use $\SU(N)$ to denote the special unitary group of $N\times N$ matrices. Let $\Haar$ denote the Haar measure on $\SU(N)$. The Haar measure is the unique measure on $\SU(N)$ that is invariant under right-multiplication and left-multiplication by $\SU(N)$.  

\subsection{Representation Theory of \texorpdfstring{$\SU(N)$}{SU(N)}}

We now describe representations of groups. For a group $G$, a representation $(\pi,V)$ of $G$ is a \emph{group homomorphism} $\pi:G\to \Gcal\Lcal(V)$, i.e., a map from $G$ to non-singular complex matrices satisfying $\pi(gh)=\pi(g)\pi(h)$ for all $g,h\in G$ and $\pi(1)=\id$.  Throughout this paper, we will assume that  our group $G$ will be compact and  our representations are finite dimensional. For notational convenience, we refer to the representation $(\pi,V)$ simply as $V$ or as $\pi$. We also abuse notation by writing $\pi(g)v$ as $gv$ for $g\in G, v\in V$. For any finite dimensional representation $(\pi,V)$, there is a basis for $V$ according to which $\pi(g)$ is unitary for all $g\in G$ and we will typically work with such a basis. A $G$-\emph{morphism} between irreps is a map   $\varphi\colon V\to U,$ satisfying $\varphi(gv) = g\varphi(v)$ (i.e., $\varphi(\pi(g)v)=\pi(g)\varphi(v)$.) for all $g\in G$ and $v\in V$

\paragraph*{Irreducible representations and Schur's Lemma}

Two representations $(\pi, V), (\rho, U)$ are said to be \emph{isomorphic} (which we denote by $\pi \sim \rho$) if there exists an invertible $G$-morphism between them. Otherwise they are \emph{non-isomorphic} (denoted $\pi \nsim \rho$).   We denote $U\le V$ to be a \emph{subrepresentation} if $gu\in U$ for all $g\in G$ and $u\in U$. A representation is said to be an \emph{irreducible representation} (or \emph{irrep}) if its only subrepresentations are 0 and itself, i.e., if it cannot be decomposed as the direct sum of two non-trivial representations. We use $\widehat{G}$ to denote a complete set of irreps of $G$, that is, every irrep of $G$ is isomorphic to some irrep in $\widehat{G}$. We will also make use of Schur's lemma which we state below. 
\begin{lemma}[Schur's lemma]
    Let $(\pi,V),(\rho,U)$ be two irreps of $G$. Let $\varphi$ be a $G$-morphism between $\pi$ and $\rho$. If $\pi\nsim \rho$, then $\varphi$ is $0$ and if $\pi=\rho$, then $\varphi$ is a scalar multiple of~identity. 
\end{lemma}

\paragraph*{Matrix Coefficients}
Let $(\pi,V)\in \widehat{G}$ be an irrep of $G$. We use $M_\pi\subseteq L^2(G)$ to denote the space spanned by functions $f_{u,v}:G\to \Rbb$ of the form $g\to \langle u,\pi(g) v\rangle$ for $u,v\in V$. We refer to $M_\pi$ as the space of \emph{matrix coefficients} associated to the representation $\pi$. For any $i,j\in [\dim(V)]$, let $\widetilde{\pi}_{i,j}$ be $\sqrt{\dim(V)} \pi_{i,j}$ where $\pi_{i,j}$ is the function defined as 
$$
{\pi}_{i,j}:g \mapsto  \pi(g)_{i,j} = \langle e_i,\pi(g) e_j\rangle
$$ for all $g\in G$. 
We now state  Schur's orthogonality relation, a corollary of Schur's~lemma.
\begin{fact}[Schur's Orthogonality Relations] \label{fact:schurs_orthogonality}
	Let $(\pi,V),(\sigma,W)\in \widehat{G}$.  Then, 
 $$
 \E_g\sbra{ \widetilde{\pi}(g)_{i,j}\cdot  \overline{\widetilde{\sigma}(g)_{k,\ell}}  }=\indi[\sigma= \pi, i=k,j=\ell], \qquad \forall \hspace{1mm} i,j\in [\dim(V)],k,\ell \in [\dim(W)].
 $$
\end{fact}
In other words,  $\{ \widetilde{\pi}_{i,j}: i,j\in \dim(V)\} $ is an orthonormal basis for $M_\pi$ and,  $M_\pi,M_\sigma$ are orthogonal for $\pi\nsim \sigma$.

\paragraph*{Peter-Weyl Theorem} 
The Peter Weyl theorem states that the space of all matrix coefficients is dense in $L^2(G)$. In other words, $L^2(G)$ can be decomposed as an orthogonal direct sum  of matrix coefficients  $\{M_\pi:\pi \in \widehat{G}\}$. 
\begin{theorem}[Peter-Weyl Theorem]
If $G$ is a group equipped with the Haar measure, then
$$
L^2(G)={\underset{\pi\in \widehat{G}}{\bigoplus}}M_\pi.
$$
Furthermore, every closed subspace of $W\subseteq L^2(G)$,  that commutes
with the action of $G$ from both sides can be written as $W={\underset{\pi\in L_W}{\bigoplus}}M_\pi$ for some $L_W\subseteq \widehat{G}$.
\end{theorem}
This provides a very natural basis to study $L^2(G)$ and an analogue of Fourier analysis for $G$. 

\subsection{Fourier Coefficients}
For any function $f\in L^2(G)$ and for any $(\pi,V)\in \widehat{G}$, define the Fourier coefficient  $\widehat{f}(\pi)\in \Lcal(V)$ as 
$$
\widehat{f}(\pi)=\E_g[ f(g)\cdot  \widetilde{\pi}(g^{-1})],
$$
where the expectation is with respect to the Haar measure on $G$. 
The Peter-Weyl theorem implies that the Fourier decomposition of $f$ can be written as
$$
f(g)=\underset{\substack{\pi \in \widehat{G}\\i,j\in [\dim(\pi)]}}{\sum} \widehat{f}(\pi)_{i,j}\cdot \widetilde{\pi}(g)_{j,i}
$$
for all $g\in G$. Define $\|f\|_2^2=\E_g[|f(g)|^2]$. Similarly to classical Boolean function analysis, one can define the Plancharel's theorem and convolutions of functions in $L^2(G)$ which we describe now.
\begin{theorem}[Plancharel's Theorem]
For every $f,h\in L^2(G)$, we have that
$$
\E_g[f(g)\overline{h(g)}]=\underset{\substack{\pi \in \widehat{G}\\i,j\in [\dim(\pi)]}}{\sum} \widehat{f}(\pi)_{i,j} \cdot \overline{\widehat{h}(\pi)_{i,j}}.
$$
In particular, $\E_g[|f(g)|^2]=\underset{\substack{\pi \in \widehat{G}\\i,j\in [\dim(\pi)]}}{\sum} \abs{\widehat{f}(\pi)_{i,j}}^2$.
\end{theorem}
This follows from the following calculation.
\begin{align*}
\E_g[f(g)\overline{h(g)}]&=\underset{\substack{\pi,\sigma \in \widehat{G}\\i,j\in [\dim(\pi)]\\k,\ell\in [\dim(\sigma)]}}{\sum} \E_g\sbra{\widehat{f}(\pi)_{i,j}\cdot \widetilde{\pi}(g)_{j,i}\cdot \overline{\widehat{h}(\sigma)_{k,\ell}}\cdot \overline{\widetilde{\sigma}(g)_{\ell,k}}}\\
&=\underset{\substack{\pi,\sigma \in \widehat{G}\\i,j\in [\dim(\pi)]\\k,\ell\in [\dim(\sigma)]}}{\sum} \widehat{f}(\pi)_{i,j}\cdot \overline{\widehat{h}(\sigma)_{k,\ell}}\cdot \E_g[\widetilde{\pi}(g)_{j,i}\cdot \overline{\widetilde{\sigma}(g)_{\ell,k}}]
=\underset{\substack{\pi \in \widehat{G}\\i,j\in [\dim(\pi)]}}{\sum} \widehat{f}(\pi)_{i,j} \cdot \overline{\widehat{h}(\pi)_{i,j}}.
\end{align*}  

\paragraph*{Convolution of Functions}
For every $f_1,f_2\in L^2(G)$, define their convolution  $f_1*f_2\in L^2(G)$~by 
\[
(f_1*f_2)(g)=\E_h\sbra{f_1(gh^{-1})f_2(h) }
\] 
for all $g\in G$. We make use of the following formula for the convolutions of two elements of our orthonormal basis of $L^2(G)$.


\begin{fact} \label{fact:convolution}
Let $(\pi,V),(\sigma,W) \in \widehat{G}$. For all $i,j\in [\dim(V)], k,\ell\in[\dim(W)]$, we have that
\[
\widetilde{\pi}_{i,j}*\widetilde{\sigma}_{k,\ell}=\frac{\indi[j=k, \pi\sim \sigma]}{\sqrt{\dim(V)}}\cdot \widetilde{\pi}_{i,\ell}
\]
\end{fact}

\subsection{Hypercontractivity on \texorpdfstring{$\SU(N)$}{SU(N)}}
Consider the  group $G=\SU(N)$. For $X\in \SU(N)$ and $d\geq 1$, define $V_{\leq d}$ to consist of functions representable as degree $d$  multilinear polynomials in the formal variables $\{\textsf{Re}(X_{ij}),\textsf{Im}(X_{ij}):i,j\in [N]\}$ where $X\in \SU(N)$. For every $d\in \{0,\ldots,N/2-1\}$, define $V_{=d}:= V_{\leq d}\cap (V_{\leq d-1})^{\perp}$ as the ``degree-$d$" part and $V_{\geq N/2}=(V_{<N/2})^{\perp}$.
The space $L^2(G)$ can be decomposed as
$$
\bigoplus_{d=0}^{N/2-1}V_{=d} \bigoplus V_{\ge N/2}.
$$
 Furthermore, since $V_{=d}\subseteq L^2(G)$ is closed  and commutes with the action of $G$ from both sides (i.e., linear combination of degree-$d$ polynomials remains degree-$d$), by the Peter-Weyl theorem we~have
$$
V_{=d}=\bigoplus_{\pi \in L_d}M_{\pi}
$$
for some $L_d\subseteq \widehat{G}$. One important contribution of~\cite{EKLM23} is in proving important properties of this decomposition, which we discuss now. For every $f\in L^2(G)$ and $0\le d\le N/2-1$, let $f^{=d}$ denote the projection of $f$ onto $V_{=d}$, i.e., $f^{=d}=\arg\min\{\langle g,f\rangle: g\in V_{=d}\}$. Let $Q_d=\min_{\pi \in L_d} \{\dim(\pi)\}$ be the minimal dimension of any non-trivial subrepresentation of $V_{=d}$. 
 Let $f^{\ge N/2}$ denote the projection of $f$ onto $V_{\ge N/2}$ and $Q_{\ge N/2}$ be the minimal dimension of any representation of $V_{\ge N/2}$. With this, we are now ready to state the main results we use from~\cite{EKLM23}.

\begin{theorem}[Implied by {\cite[Theorems~3.7, 4.5]{EKLM23}}]
\label{theorem:level_k_inequality} There exists universal constants $c,C>0$ such that the following holds. Let $f:\SU(N)\times\SU(N)\to \{0,1\}$,  $\alpha=\E[f]$ and $d\in \Nbb$ such that $d\le \min\{c\sqrt{N},\log(1/\alpha)/2\}$. Then
\[
\vabs{f^{=d}}_2^2 \le \pbra{C/d}^d   \alpha^2\log^d(1/\alpha). 
\]
\end{theorem}
Furthermore, one can bound the dimension of the irreps occuring in $V_{=d}$ as follows.

\begin{theorem}[{\cite[Theorem~3.3]{EKLM23}}]
\label{theorem:minimal_irreps} Let $G=\SU(N)$, $d\leq N/2-1$. Let $c>0$ be a universal~constant. Then
\begin{itemize}
    \item  $Q_d\ge \pbra{\frac{cN}{d}}^d$ if $d< cN/(1+c)$,
    \item $Q_d\ge (1+c)^{cN/(1+c)}$ if $d\ge cN/(1+c)$.
\end{itemize} 
Furthermore, every irrep of $V_{\ge N/2}$ has dimension at least $(1+c)^{cN/(1+c)}$.
\end{theorem}

\section{Quantum Models \& Quantum Upper Bound}
\label{sec:main_proof_quantum}
We begin by describing the quantum models of communication and then presenting the upper~bounds. 

\subsection{Clean-Qubit Model} \label{sec:clean_qubit_model}
This model was first defined by Klauck and Lim in~\cite{KL19}. A $k$-clean-qubit quantum protocol consists of $k$ qubits in the state $\ket{0}$ and $m$ qubits that are unentangled from these and in the totally mixed state. There is no other private memory for the players. The players communicate as in a standard quantum protocol, that is, they apply unitary operators on the $m+k$ qubits and exchange them back and forth. At the end of the computation, an arbitrary projective measurement (independent of the inputs) is performed. The outcome of the measurement is declared as the output of the protocol. As in standard quantum protocols, the cost of such a protocol is the total number of qubits exchanged, which in this case is the number of rounds times $(m+k)$. Note that the clean and mixed qubits are allowed to have correlations between them, and the clean qubit can act as a control over the mixed qubits.

\subsection{Quantum Simultaneous with Entanglement}
\label{sec:quantum_model}
In this model, in addition to Alice and Bob, there is a third party Charlie. Alice and Bob initially share an entangled state (that is independent of their inputs) and each apply a quantum channel (dependent on their inputs) on their part of the entangled state and send all the qubits to Charlie. Charlie applies a projective measurement (independent of the inputs) is performed. The outcome of the measurement is declared as the output of the protocol. The cost of the protocol is the total number of qubits sent to Charlie.

\subsection{Entangled Fingerprinting Model} 
\label{sec:entangled_fingerprinting}
An entangled-fingerprinting protocol is a simple type of quantum simultaneous protocol with entanglement where essentially, Charlie  just performs a swap test. In more detail, Alice and Bob use their entanglement to prepare a state of the form
\[\frac{1}{\sqrt{2N}}\sum_{i\in [N]}\pbra{\ket{0_A,0_B}\ket{u_i}_A\ket{v_i}_B+\ket{1_A,1_B}\ket{u_i'}_A\ket{v_i'}_B},
\]  
where $\ket{u_i}_A,\ket{u'_i}_A$ are quantum states prepared by Alice and $\ket{v'_i}_B,\ket{v_i}_B$ are states prepared by Bob and the index denotes the player to which the qubit belongs. The players send this entire state to Charlie. Charlie first ``uncomputes'' the second qubit to obtain
\[
\frac{1}{\sqrt{2N}}\sum_{i\in [N]}\pbra{\ket{0}\ket{u_i}\ket{v_i}+\ket{1}\ket{u_i'}\ket{v_i'}},
\]  
and then performs a swap test on this state. In more detail, she swaps the last few  registers controlled on the first qubit to obtain
\[
\frac{1}{\sqrt{2N}}\sum_{i\in [N]}\pbra{\ket{0}\ket{u_i}\ket{v_i}+\ket{1}\ket{v_i'}\ket{u_i'}},
\]  
and then measures the first qubit in the Hadamard basis and returns 1 iff the outcome is $\ket{+}$. 

The motivation for calling this the ``entangled-fingerprinting'' model is as follows. In the standard quantum fingerprinting model of communication~\cite{GKW06}, Alice and Bob send quantum states $\ket{u}$ and $\ket{v}$ respectively to Charlie, who then performs a swap test and returns $1$ with probability $\tfrac{1}{2}\pbra{1+\braket{u|v}^2}$. It was shown~\cite{GKW06} that this model is efficiently simulable in the classical randomized simultaneous model of communication. The entangled-fingerprinting model can be viewed as a variant of the standard quantum fingerprinting model where Alice and Bob are allowed to share entanglement. In contrast to the standard fingerprinting model, we show that the entangled-fingerprinting model can exponentially outperform even interactive classical~communication.

\subsection{Quantum Upper Bound with One Clean Qubit}
\label{sec:clean_qubit_upper_bound} As we mentioned in the introduction, the $\ABCD$ problem was shown to have a simple quantum communication protocol with $O(\log N)$ qubits of communication using one clean qubit~\cite{KL19}. For completeness, we include a proof of this here. 
\begin{theorem}
There is a quantum protocol of cost $O(\log N)$ for the \ABCD~problem in the one-clean-qubit model such that \textsc{yes} instances are accepted with probability at least $0.95$ and the \textsc{no} instances are accepted with probability at most $0.55$. \label{thm:clean_qubit_upper_bound}
\end{theorem} 
This gap of $0.5$ between \textsc{yes} and \textsc{no} instances can be amplified to an arbitrary constant by repeating the protocol $O(1)$ times (and using $O(1)$ clean qubits). 

\begin{proof}[Proof of~\Cref{thm:clean_qubit_upper_bound}]
Consider the (mixed) state on $\log N+1$ qubits identified by the density matrix $\tfrac{1}{N} \begin{bmatrix} \id & 0 \\ 0 & 0 \end{bmatrix}.$ This state can be viewed as a probability mixture over pure states $\ket{0}\otimes \ket{v}$ where $v\in \Cbb^N$ is a uniformly random unit vector. The protocol starts by Alice first applying the Hadamard operator on the first qubit to produce the uniform mixture over $\frac{1}{\sqrt{2}}\pbra{\ket{0}+\ket{1}}\otimes \ket{v}$ over a random unit vector $v\in \Cbb^N$. Alice applies $A^\dagger$ controlled on the first qubit being $\ket{1}$ and sends the entire state to Bob, who applies $B^\dagger$ controlled on the first qubit being one and sends it to Alice. They similarly apply $C^\dagger$ and $D^\dagger$. This produces the uniform mixture over 
\[ \frac{1}{\sqrt{2}}\ket{0}\otimes \ket{v}+\frac{1}{\sqrt{2}}\ket{1}\otimes D^\dagger C^\dagger B^\dagger A^\dagger \ket{v}\]
for a random unit vector $v\in \Cbb^N$. Alice now measures the first qubit in the Hadamard basis and returns $\textsc{yes}$ if and only if the outcome is $\ket{+}$. The probability of outcome $\ket{+}$ is precisely 
\begin{align*}
&\frac{1}{4}\vabs{ \ket{v}+D^\dagger C^\dagger B^\dagger A^\dagger \ket{v} }_2^2\\
&=\frac{1}{4}\big(\bra{v}+ \bra{v}ABCD\big) \big( \ket{v}+D^\dagger C^\dagger B^\dagger A^\dagger \ket{v}\big)\\
&= \frac{1}{2}+\frac{1}{4}\pbra{\braket{v|ABCD|v}  + \braket{v|(ABCD)^\dagger|v} }.
\end{align*}
If $\Tr(ABCD)\ge 0.9N$, then the average of the above quantity (over a random unit vector $v$) quantity is at least $0.95$ and if $\Tr(ABCD)\le 0.1N$, then the average of this quantity at most~$0.55$. 
\end{proof}

\subsection{Quantum Upper Bound with Entangled Fingerprints}

\label{sec:quantum_upper_bound}
In this section, we show that the $\ABCD$ problem can be solved in the entangled-fingerprinting model in the $\SMP$ communication model.
\begin{theorem}\label{thm:quantum_upper_bound}
    There is a quantum protocol of cost $O(\log N)$ in the entangled-fingerprinting model for the $\ABCD$ problem when Alice and Bob share $\Theta(\log N)$ EPR pairs such that the \textsc{yes} instances are accepted with probability  $\geq 0.95$ and the \textsc{no} instances are accepted with probability~$\leq 0.55$.
\end{theorem} 

\begin{proof}[Proof of \Cref{thm:quantum_upper_bound}]
The protocol is as follows. We express the initial state shared by Alice and Bob as follows. 
\begin{align*}
\frac{1}{\sqrt{2N}} \Big(\ket{0_A0_B}\sum_{i=1}^N\ket{i_A,i_B}+\ket{1_A1_B}\sum_{i=1}^N\ket{i_A,i_B}\Big)
\end{align*}
where the subscript $A,B$ denote the registers that Alice and Bob have respectively. Alice applies the map $\begin{bmatrix} A & 0 \\ 0 & C\end{bmatrix}$ which maps $\ket{0_A}\ket{i_A}$ to $\ket{0_A}A\ket{i_A}$ and $\ket{1_A}\ket{i_A}$ to $\ket{1_A}C \ket{i_A}$. Bob applies the map $\begin{bmatrix} B^\dagger & 0 \\ 0 & D^\dagger \end{bmatrix}$ which maps $\ket{0_B}\ket{i_B}$ to $\ket{0_B}B^\dagger\ket{i_B}$ and $\ket{1_B}\ket{i_B}$ to $\ket{1_B}D^\dagger\ket{i_B}$.  This produces the state
\begin{align*}
\frac{1}{\sqrt{2N}} \Big( \ket{0_A,0_B}\sum_{i=1}^N A\ket{i_A}B^\dagger\ket{i_B}+\ket{1_A1_B}\sum_{i=1}^N C \ket{i_A}D^\dagger\ket{i_B}\Big).
\end{align*}
They send the above state  to Charlie. Charlie first uncomputes the second qubit to obtain
\begin{align*}
\frac{1}{\sqrt{2N}} \Big( \ket{0}\sum_{i=1}^N A\ket{i}B^\dagger\ket{i}+\ket{1}\sum_{i=1}^N C \ket{i}D^\dagger\ket{i}\Big).
\end{align*}
Charlie then does a swap test on this state. That is, she applies a controlled swap between the last two registers controlled on the first register and then measures the first qubit in the Hadamard basis. 
If the outcome is $\ket{+}$ then she outputs \textsc{yes}, else outputs \textsc{no}. The probability of outcome $\ket{+}$ is  precisely
\begin{align*}
& \frac{1}{8N}\vabs{ \sum_{i=1}^N A\ket{i}B^\dagger \ket{i} + \sum_{i=1}^N D^\dagger\ket{i}C \ket{i}}_2^2\\
&=\frac{1}{{8N}} \Big( \sum_{j=1}^N \bra{j}A^\dagger\bra{j}B+\sum_{j=1}^N\bra{j}D\bra{j}C^\dagger \Big)  \Big( \sum_{i=1}^N A\ket{i}B^\dagger \ket{i}+\sum_{i=1}^N D^\dagger\ket{i}C\ket{i}\Big)\\
&= \frac{1}{{8N}} \Big(2N+ \sum_{i,j=1}^N \bra{j}A^\dagger D^\dagger\ket{i}\bra{j}BC\ket{i}+\sum_{i,j=1}^N\bra{j}DA\ket{i}\bra{j}C^\dagger B^\dagger \ket{i}\Big)\\
&= \frac{1}{{8N}} \Big(2N+ \sum_{i,j=1}^N \bra{j}A^\dagger D^\dagger \ket{i}\bra{i}C^\dagger B^\dagger\ket{j}+\sum_{i,j=1}^N\bra{j}DA\ket{i}\bra{i}BC\ket{j}\Big)\\
&= \frac{1}{4}+\frac{1}{{8N}} \Big( \Tr\big( (ABCD)^\dagger)+\Tr\big( ABCD\big)\Big)
\end{align*}
If $\Tr(ABCD)\ge 0.9N$, then the above quantity is at least $0.475$ and the protocol is correct with probability at least $0.95$ whereas if $\Tr(ABCD)\le 0.1N$, the above quantity is at most $0.275$ and the protocol returns \textsc{yes} with probability at most $0.55$. \end{proof}

\section{Classical Lower Bound}
\label{sec:classical_lower_bound}
In this section we will prove the following theorem.
\begin{theorem}\label{theorem:main_theorem}
The randomized communication complexity of the $\ABCD$ problem is $\Omega(\sqrt{N}).$
\end{theorem}

The main technical lemma of this paper is the following. 
\begin{lemma}\label{lemma:main_lemma}
Let $f,g:\SU(N)\times \SU(N)\to \{0,1\}$ be such that $\E[f],\E[g]\ge e^{-c'\sqrt{N}}$ for a sufficiently small global constant $c'>0$. Then,
\[
\abs{\E[f(A,C)g(B,(ABC)^{-1})]-\E[f(A,C)g(B,D)]}\le \E[f(A,C)g(B,D)]/30,
\]
where all the expectations are taken with respect to the Haar measure over $\SU(N)$.
\end{lemma}

The proof of the main theorem essentially follows by adding this inequality over all rectangles in the protocol. We first prove the theorem assuming the lemma. 
\begin{proof}[Proof of~\Cref{theorem:main_theorem} from~\Cref{lemma:main_lemma}]
Let $c'$ be the global constant as in~\Cref{lemma:main_lemma}.  We will show that any randomized communication protocol of cost $c'\sqrt{N}/2$  has advantage at most $1/10$ in distinguishing \textsc{yes} and \textsc{no} instances of the $\ABCD$ problem. To this end, consider the \emph{hard} distributions defined as follows. 
\begin{definition}[Hard Distributions] For the $\textsc{yes}$ distribution, Alice gets  $A,C\sim \Haar$, Bob gets  $B,D\sim \Haar$. For the $\textsc{no}$ distribution, Alice gets $A,C\sim \Haar$ and Bob gets $B\sim \Haar$ and $D=(ABC)^{-1}$. 
\end{definition}
Let $\mathcal{P}:\SU(N)^4\to\{0,1\}$ denote the (probabilistic) output of any such randomized communication protocol. Here, we view the output of the protocol as a probabilistic bit in $\{0,1\}$. By Yao's lemma and the triangle inequality, it suffices to bound the difference 
\begin{equation}\label{equation:difference}
\E[\mathcal{P}(A,B,C,(ABC)^{-1})]-\E[\mathcal{P}(A,B,C,D)]\end{equation}
 for \emph{deterministic} protocols. Fix any deterministic protocol $\Pcal:\SU(N)^4\to\{0,1\}$ of cost $c'\sqrt{N}/2$. This defines a partition of the input space into rectangles, where a typical rectangle has  measure $2^{-c'\sqrt{N}/2}$ under the \textsc{no} distribution. Let $\Rcal$ denote the set of rectangles of \textsc{no}-measure at least $2^{-c'\sqrt{N}}$. Observe that \textsc{no}-measure of $\Rcal$ is at least $1-2^{c'\sqrt{N}/2}\cdot 2^{-c'\sqrt{N}}\ge 1- 2^{-c'\sqrt{N}/2}$. We will analyze the contribution of each rectangle in Eq.~\eqref{equation:difference} separately, based on whether it is in $\Rcal$ or not. 
 
 Fix any rectangle $f\times g$ in $\Rcal$ where $f,g:\SU(N)^2\to\{0,1\}$ are the indicator functions of Alice's and Bob's sets. Since the \textsc{no}-measure of the rectangle is precisely $\E[f]\cdot \E[g]$, we have $\E[f],\E[g]>e^{-c'\sqrt{N}}$. We now apply~\Cref{lemma:main_lemma} to conclude that
\begin{equation} \label{equation:inequality}
\abs{\E[f(A,C)g(B,(ABC)^{-1})]-\E[f(A,C)g(B,D)]}\le \E[f(A,C)g(B,D)]/30.
\end{equation}
We already argued that the \textsc{no}-measure of $\Rcal$ is at least $1-2^{-c'\sqrt{N}}$. We now add Eq.~\eqref{equation:inequality} over all rectangles in $\Rcal$ to conclude that the  \textsc{yes}-measure of $\Rcal$ is at least $1-2^{-c'\sqrt{N}/2} -1/30\ge 1-1/20$. In particular, the total \textsc{yes}-measure of rectangles \emph{not in} $\Rcal$ is at most $1/20$. Hence, the total contribution of such rectangles to Eq.~\eqref{equation:difference} is at most $1/20+2^{-c'\sqrt{N}/2}$. We now consider the contribution of rectangles in $\Rcal$. We again use Eq.~\eqref{equation:inequality} and add up over all one-rectangles in $\Rcal$ to conclude that the total contribution of such rectangles to Eq.~\eqref{equation:difference} is at most $1/30$. Overall, we have
\[
\E[\mathcal{P}(A,B,C,(ABC)^{-1})]-\E[\mathcal{P}(A,B,C,D)]\le 1/30+1/20+2^{-c'\sqrt{N}/2} <1/10.
\]
Here, we used the facts that the protocol is constant within a rectangle and the output of the protocol is in $\{0,1\}$. This completes the proof of~\Cref{theorem:main_theorem}.
\end{proof}

It remains to prove~\Cref{lemma:main_lemma} which we do in the next section. 


\subsection{Proof of~\texorpdfstring{\Cref{lemma:main_lemma}}{Lemma 5.2}}\label{sec:main_lemma}
In order to prove this lemma, we need two claims: the first one decomposes  the main expression we need to bound in Lemma~\ref{lemma:main_lemma} in terms of the Fourier coefficients and the next claim is a corollary of Theorem~\ref{theorem:level_k_inequality}.
\begin{claim} 
\label{claim:fourier_expansion} 
Let $G=\SU(N)$ and let $f,g:G\times G\to \{0,1\}$ be the indicator functions. Then, 
$$
\E\sbra{f(A,C)g(B,(ABC)^{-1})}=\sum_{\pi\in \widehat{G}}\frac{1}{\dim(\pi)} \sum_{\substack{i,j\in[\dim(\pi)]\\k,\ell\in [\dim(\pi)]}}  {\widehat{f}(\pi,\pi)_{k,j,\ell,i}}\cdot {\widehat{g}(\pi,\pi)_{i,k,j,\ell}}.
$$
\end{claim}
\begin{proof}
    This follows by expanding $f,g$ in the Fourier basis. To simplify the analysis, we introduce the diagonal probability measure $\mud$. It is the measure obtained by sampling $x\sim \SU(n)$ and outputting $(x,x^{-1})$. More formally it is the push-forward of the Haar measure with respect to the map $x\mapsto (x,x^{-1})$. Observe that
\begin{align*} 
	\E\sbra{f(A,C)g(B,(ABC)^{-1})} &= \E\sbra{f(A,C)g(B,D) \mid AB=(CD)^{-1} } \\
	&= \E\sbra{(f*g)(AB,CD) \mid AB=(CD)^{-1} } \\
	&= \langle \mud,f*g\rangle, 
\end{align*}
where the second equality used the definition of convolution to get $\E[f(A,C)g(B,(ABC)^{-1}]=\E[(f*g)(AB,(AB)^{-1})]$.
We can identify $\widehat{G\times G}$ with the tensor product $\widehat{G}\otimes \widehat{G}$ and accordingly, we will index the irreps of $G\times G$ by $\pi\otimes \sigma$ for $\pi,\sigma\in \widehat{G}$ and we refer to the corresponding Fourier coefficients of a function $h:G\times G\to \Rbb$ by~$\widehat{h}(\pi,\sigma)$. Using Plancharel's theorem, we now rewrite the above as
\begin{align} 
\label{eq:expandingmudabove}
\langle \mud,f*g \rangle &= \sum_{\pi,\sigma\in \widehat{G}} \sum_{\substack{i,j\in[\dim(\pi)]\\k,\ell\in [\dim(\sigma)]}} \overline{\widehat{\mud}(\pi,\sigma)_{i,j,k,\ell}}\cdot \widehat{f*g}(\pi,\sigma)_{i,j,k,\ell}
\end{align}
We now compute the Fourier coefficients of $\mud$ as well as those of $f*g$. We first write out the Fourier coefficients of $\mud$. 
Let $\pi,\sigma\in \widehat{G}$ and let $i,j\in [\dim(\pi)],k,\ell\in [\dim(\sigma)]$. We have 
\begin{align*}
\widehat{\mud}(\pi,\sigma)_{i,j,k,\ell}&=\E\sbra{\widetilde{\pi}_{i,j}(X)\widetilde{\sigma}_{k,\ell}(X^{-1})}=\E\sbra{\widetilde{\pi}_{i,j}(X)\overline{\widetilde{\sigma}_{\ell,k}(X)}},
\end{align*}
where we used $\widehat{\mud}(\pi,\sigma)=\E_{g_1,g_2}[\mud(g_1,g_2)\pi(g_1^{-1})\otimes \sigma(g_2^{-1})]$ in the first equality and $\sigma(X^{-1})=\overline{\sigma(X)^T}$ in the second equality. 
We now use Schur's orthogonality relations in \Cref{fact:schurs_orthogonality} to conclude that the RHS above is $1$ if $\pi=\sigma$, $i=\ell$, $j=k$ and $0$ otherwise, hence we get that
\begin{align}
\label{eq:mudcoeff}
\widehat{\mud}(\pi,\sigma)_{i,j,k,\ell}= \begin{cases} 
      0 & \pi\neq \sigma \\
     \indi[i=\ell,j=k] & \pi=\sigma.
   \end{cases}
\end{align}
So it suffices to consider the terms $\pi=\sigma$ and $i=\ell,j=k$ in Eq.~\eqref{eq:expandingmudabove}.  We next write out the Fourier coefficients of $\widehat{(f*g)}(\pi,\pi)$. 
Similar to the convolution property in~\Cref{fact:convolution}, we have that
\begin{align*}
&\widehat{(f*g)}(\pi,\pi)_{i,j,j,i}\\
&= \E\sbra{(f*g)(X,Y)\widetilde{\pi}_{i,j}(X^{-1})\widetilde{\pi}_{j,i}(Y^{-1})}\\
&=\E\sbra{f(A,C)g(B,D) \widetilde{\pi}_{i,j}(B^{-1}A^{-1})\widetilde{\pi}_{j,i}(D^{-1}C^{-1})}\\
&=\dim(\pi)\cdot \E\sbra{f(A,C)g(B,D) \cdot \pi(B^{-1}A^{-1})_{i,j}\cdot \pi(D^{-1}C^{-1})_{j,i}} \\
 &=\frac{1}{\dim(\pi)}\cdot \E\sbra{f(A,C)g(B,D) \pbra{\sum_{k\in[\dim(\pi)]}\widetilde{\pi}_{i,k}(B^{-1})\widetilde{\pi}_{k,j}(A^{-1})}\pbra{\sum_{\ell\in [\dim(\pi)]}\widetilde{\pi}_{j,\ell}(D^{-1})\widetilde{\pi}_{\ell,i}(C^{-1})}}\\
 &=\frac{1}{\dim(\pi)}\cdot \sum_{k,\ell\in[\dim(\pi)]}\widehat{f}(\pi,\pi)_{k,j,\ell,i}\cdot \widehat{g}(\pi,\pi)_{i,k,j,\ell}.
\end{align*}
Putting together the above equality and Eq.~\eqref{eq:mudcoeff} into Eq.~\eqref{eq:expandingmudabove}, we get that 
\begin{align*} 
\langle \mud,f*g \rangle &= \sum_{\pi,\sigma\in \widehat{G}} \sum_{\substack{i,j\in[\dim(\pi)]\\k,\ell\in \dim(\sigma)}} \overline{\widehat{\mud}(\pi,\sigma)_{i,j,k,\ell}}\cdot {\widehat{f*g}(\pi,\sigma)_{i,j,k,\ell}}\\
 &= \sum_{\pi\in \widehat{G}} \sum_{i,j\in[\dim(\pi)]} {\widehat{f*g}(\pi,\pi)_{i,j,j,i}}= \sum_{\pi\in \widehat{G}} \frac{1}{\dim(\pi)}\sum_{\substack{i,j\in[\dim(\pi)]\\k,\ell\in [\dim(\pi)]}} {\widehat{f}(\pi,\pi)_{k,j,\ell,i}}\cdot {\widehat{g}(\pi,\pi)_{i,k,j,\ell}},
\end{align*}
hence proving the claim statement.
\end{proof}

We next prove the following claim which follows from~\Cref{theorem:level_k_inequality} and holds for all $d\in \Nbb$.
\begin{claim}\label{corollary:level_k_inequality}
 There exists global constants $c,C>0$ such that the following holds. Let  $f:\SU(N)\times\SU(N)\to \{0,1\}$ and $\alpha=\E[f]$ be such that $\alpha \ge e^{-c\sqrt{N}}$. Then, for all $d\in\Nbb$, we have
\[
\vabs{f^{=d}}_2^2 \le (6C)^d \alpha^2 + \big(2C/d\cdot  \log(1/\alpha)\big)^d\alpha^2
\]
\end{claim}
\begin{proof}
Fix global constants $C,c$ as in~\Cref{theorem:level_k_inequality}. We show $C'=6C$ and $c'\le c$ that satisfy~\Cref{corollary:level_k_inequality}.  Let $d\le c\sqrt{N}$. Let $K=\lceil e^{2d}\rceil$. Observe that $d\le \log(K/\alpha)/2$ for all $\alpha\in[0,1]$. We express $f:G\to \{0,1\}$ as a sum of indicator functions $\{f_i:G\to \{0,1\}\}_{i\in [K]}$ where each $f_i$ satisfies $\E[f_i]=\alpha/K$. This can be done by partitioning the set corresponding to $f$ into $K$ parts, each of measure $\alpha/K$. We now apply~\Cref{theorem:level_k_inequality} to each $f_i$ of level $d$ to conclude that
\begin{align}
\label{eq:leveldinequality}
\vabs{f^{=d}_i}^2_2 \le \frac{1}{K^2 d^d} C^d \alpha^2 \log^d(K/\alpha). \end{align}
We now add this inequality for all $i\in[K]$. 
We obtain the bound for levels $d\le c\sqrt{N}$ from the following calculation.
\begin{align*} 
\vabs{f^{=d}}^2_2 &\leq K\cdot \sum_{i\in[K]} \|f_i^{=d}\|_2^2\\
&\le K^2\cdot \frac{1}{K^2d^d} C^d \alpha^2 \log^d(K/\alpha)\\
& \le \frac{C^d}{d^d}\alpha^2 \pbra{3d+\log(1/\alpha)}^d\\
&\le \frac{C^d}{d^d}\alpha^2 2^d\pbra{(3d)^d+\log(1/\alpha)^d}\le (6C)^d \alpha^2 + \big(2C/d\cdot  \log(1/\alpha)\big)^d\alpha^2,
\end{align*}
where the first inequality Cauchy-Schwarz and second inequality used Eq.~\eqref{eq:leveldinequality}.  
For levels $d> c\sqrt{N}$, first observe that $\|f^{=d}\|_2^2\le \alpha$ by Parsevals identity. Now choose $c'\le c$ to be a small enough constant  so that when $\alpha \ge e^{-c'\sqrt{N}}$ we have $\alpha\le (6C)^{c\sqrt{N}} \alpha^2$. This implies that we can upper bound $\|f^{=d}\|_2^2\leq \alpha$ by $(6C)^d \alpha^2$ for levels $d> c\sqrt{N}$. This proves the claim statement. 
\end{proof}

Using these two claims, we are now ready to prove  our main Lemma~\ref{lemma:main_lemma}. 
\begin{proof}[Proof of Lemma~\ref{lemma:main_lemma}] Consider the global constant $c$ as in~\Cref{corollary:level_k_inequality} and let $c'$ be sufficiently smaller than $c$.
For $0\le d\le N/2-1$,  let $L_d$ denote the set of all irreps of $V_{=d}$ and let $L_{N/2}$ denote the set of all the irreps of $V_{\ge N/2}$. From~\Cref{claim:fourier_expansion}, we have
\begin{align*}
    \Delta&:=\E[f(A,C)g(B,(ABC)^{-1})]-\E[f(A,C)g(B,D)] \\
    &= \sum_{\emptyset\neq \pi\in \widehat{G}}\frac{1}{\dim(\pi)} \sum_{\substack{i,j\in[\dim(\pi)]\\k,\ell\in [\dim(\pi)]}}  {\widehat{f}(\pi,\pi)_{k,j,\ell,i}}\cdot {\widehat{g}(\pi,\pi)_{i,k,j,\ell}}
\end{align*}
We now break up the contribution of various $\pi$ in the above summation depending on the $L_d$ to which they belong.
\begin{align*}
\Delta &\le \sum_{1\le d\le N/2,}\sum_{\emptyset\neq \pi\in L_d}\frac{1}{\dim(\pi)} \sum_{\substack{i,j\in[\dim(\pi)]\\k,\ell\in [\dim(\pi)]}}  {\widehat{f}(\pi,\pi)_{k,j,\ell,i}}\cdot {\widehat{g}(\pi,\pi)_{i,k,j,\ell}}
\end{align*}
Since $Q_d$ is the minimum dimension of a representation $\pi \in L_d$, we can lower bound $\dim(\pi)\geq Q_d$ above.  
Applying the Cauchy Schwarz on terms $\pi\in L_d$ we get that 
\begin{align}
\label{equation:main_inequality}
\Delta &\le \sum_{1\le d\le N/2,}\frac{1}{Q_d}\sum_{\emptyset\neq \pi\in L_d} \sum_{\substack{i,j\in[\dim(\pi)]\\k,\ell\in [\dim(\pi)]}}  \abs{ {\widehat{f}(\pi,\pi)_{k,j,\ell,i}}}\cdot \abs{{\widehat{g}(\pi,\pi)_{i,k,j,\ell}}}\\
&\le \sum_{1\le d\le N/2,}\frac{1}{Q_d}\sqrt{\sum_{\emptyset\neq \pi\in L_d} \sum_{\substack{i,j\in[\dim(\pi)]\\k,\ell\in [\dim(\pi)]}}  \abs{{\widehat{f}(\pi,\pi)_{k,j,\ell,i}}}^2}\cdot \sqrt{\sum_{\emptyset\neq \pi\in L_d} \sum_{\substack{i,j\in[\dim(\pi)]\\k,\ell\in [\dim(\pi)]}}  \abs{{\widehat{g}(\pi,\pi)_{i,k,j,\ell}}}^2}\\
&\leq \sum_{1\le d\le N/2} \frac{1}{Q_d} \vabs{f^{=2d}}_2\vabs{g^{=2d}}_2
\end{align}
We now analyze the expression above by considering the two cases $d\leq c'\sqrt{N}/2$ or $d> c'\sqrt{N}/2$. For notational convenience, let $\alpha=\E[f]$ and $\beta=\E[g]$.

\paragraph*{Contribution from levels $d\le c'\sqrt{N}/2$.} We will show that the contribution is at most $\alpha\beta/30$ by applying the Level-$k$ Inequality  in \Cref{corollary:level_k_inequality}.  Since $\alpha,\beta\ge e^{-c'\sqrt{N}}$ and $c'\ll c$, we can apply~\Cref{corollary:level_k_inequality} to $f,g$. This implies that for all $d\in \Nbb$, we have 
\[ \vabs{f^{=2d}}_2^2 \le C^{2d}\alpha^2 + \frac{C^{2d}}{(2d)^{2d}} \alpha^2 (c'\sqrt{N})^{2d}\quad\text{ and }\quad \vabs{g^{=2d}}_2^2 \le C^{2d}\beta^2 + \frac{C^{2d}}{(2d)^{2d}}\cdot \beta^2 (c'\sqrt{N})^{2d}. \]
Since $2d\le c'\sqrt{N}$, the second term in the above inequalities dominates. We now use the fact that $Q_d\ge (cN/d)^d$ from~\Cref{theorem:minimal_irreps}. Thus, we have 
\[ \frac{1}{Q_d} \vabs{f^{=2d}}\cdot \vabs{g^{=2d}} \le \frac{d^d}{(cN)^d}\cdot \alpha \beta\cdot 2 C^{2d} \frac{(c'\sqrt{N})^{2d}}{(2d)^{2d}} \le \alpha \beta /(50)^d \]
since $c'$ is a sufficiently small constant. Thus, the contribution from the levels $d\le c'\sqrt{N}$ is at most $\alpha\beta\cdot \sum_d 50^{-d}\le \alpha \beta/30$.

\paragraph*{Contribution from levels $d> c'\sqrt{N}/2$.}  We will show that the contribution is at most $2\alpha\beta/30$. For these levels, we use the trivial bound $\|f^{=d}\|_2 \le \sqrt{\alpha},\|g^{=d}\|_2 \le \sqrt{\beta}$ from Parseval's identity. We will need to handle $d\le cN/10$ and $d>cN/10$ separately. For $d\le cN/10$, ~\Cref{theorem:minimal_irreps} implies that $Q_d \ge (cN/d)^d$ and since $\alpha\beta \ge e^{-c'\sqrt{N}}\gg 10^{-c'\sqrt{N}}$, we have 
\[
\frac{1}{Q_d}\sqrt{\alpha\beta}  \le \sqrt{\alpha \beta} \cdot \pbra{\frac{d}{cN}}^d \le \sqrt{\alpha \beta}\cdot \pbra{\frac{1}{10}}^{c'\sqrt{N}}\le \alpha\beta/(30N).
\]
For $d\ge cN/10$, ~\Cref{theorem:minimal_irreps} implies that\footnote{ This is because for levels $d\le cN/(1+c)$, we have $(cN/d)^d \ge (1+c)^d \ge (1+c)^{cd/(1+c)}$ and for levels $d>cN/(1+c)$ we have $(1+c)^{cN/(1+c)}\ge (1+c)^{cd/(1+c)}$.} $Q_d \ge (1+c)^{cd/(1+c)}\ge e^{-\Omega(N)}$. Since $\sqrt{\alpha\beta} \ge e^{-c'\sqrt{N}}$, we have 
\[\frac{1}{Q_d}\sqrt{\alpha\beta} \le \alpha\beta/(30N). \]
The same calculation works for levels $\ge N/2$ and since $\|f^{\ge N/2}\|_2 \le \sqrt{\alpha},\|g^{\ge N/2}\|_2\le \sqrt{\beta}$, we have
\[\frac{1}{Q_{\ge N/2}} \sqrt{\alpha \beta} \le \alpha\beta/(30 N) \]
Adding this over all possible $d\ge c'\sqrt{N}/2$, it follows that the contribution from levels $d\ge c'\sqrt{N}/2$ is at most $2\alpha \beta/30$. Substituting these in Eq.~\eqref{equation:main_inequality}, we have $\Delta \le \alpha\beta/30.$
\end{proof}







\bibliographystyle{alpha}
\bibliography{refs}
\end{document}